\documentclass[aps, prl, longbibliography,twocolumn,showpacs,superscriptaddress,floatfix]{revtex4-1} 
\usepackage{color,soul}
\usepackage{csquotes}
\usepackage[export]{adjustbox}
\usepackage{dsfont,amsthm, amssymb,amsmath,stmaryrd,mathtools,
graphicx,pgfplots}
\pgfplotsset{compat=1.14}
\usepgfplotslibrary{external}
\usepackage{subfig}

\usepackage{tikz}
\usetikzlibrary{shapes.geometric}
\usetikzlibrary{positioning}
\usetikzlibrary{arrows}
\usetikzlibrary{decorations.pathmorphing}
\usetikzlibrary{calc}
\pgfmathtruncatemacro\distance{1}
\usepackage{array}

\tikzset{triangle/.style = {fill=blue!20, regular polygon, regular polygon sides=3, inner sep=0pt}, border rotated/.style = {shape border rotate=180}}
\tikzset{triangle/.style = {fill=blue!20, regular polygon, regular polygon sides=3, inner sep=0pt}}
\tikzset{hexagon/.style = {fill=blue!20, regular polygon, regular polygon sides=6, inner sep=0pt,outer sep=0pt}}
\tikzset{literal/.style={circle,fill=gray!20,draw,minimum size=1.4cm,inner sep=0pt}, }
\tikzset{hgnode/.style={circle,fill=gray!20,draw,minimum size=1.4cm,inner sep=0pt}, }
\tikzset{xgen/.style={triangle, fill=yellow!20,draw,minimum size=2.2cm,inner sep=0pt},}
\tikzset{ygen/.style={diamond, border rotated,fill=red!20,draw,minimum size=2.25cm,inner sep=0pt},}
\tikzset{ytype/.style={regular polygon,regular polygon sides=6,fill=orange!20,draw,minimum size=1.5cm,inner sep=0pt},}
\tikzset{clause/.style={rectangle, border rotated,fill=blue!20,draw,minimum width=2.2cm,minimum height=1cm, inner sep=0pt},}
\tikzset{clausetype/.style={regular polygon,regular polygon sides=5,fill=green!20,draw,minimum size=2cm,inner sep=0pt},}
\tikzset{dot/.style 2 args={fill, circle, inner sep=1pt, label={#1:\scriptsize #2}}}

\DeclareFontFamily{OT1}{pzc}{}
\DeclareFontShape{OT1}{pzc}{m}{it}{<-> s * [1.10] pzcmi7t}{}
\DeclareMathAlphabet{\mathpzc}{OT1}{pzc}{m}{it}

\usepackage{mathrsfs}
\usepackage{graphicx,color}
\usepackage{relsize}
\usepackage{url}
\usepackage{xr-hyper}
\usepackage[colorlinks=true, urlcolor=blue,citecolor=blue,anchorcolor=blue]{hyperref}
\usepackage[capitalise]{cleveref}
\usepackage{mymacros}
\usepackage{verbatim}

\newcommand{\be}{\begin{equation}}
\newcommand{\ee}{\end{equation}}
\newcommand{\ba}{\begin{align}}
\newcommand{\ea}{\end{align}}
\allowdisplaybreaks

\newcommand{\ctg}{G} 
\newcommand{\hg}{\Gamma}

\begin{document}

\title{Effective gaps are not effective: quasipolynomial classical simulation of obstructed stoquastic Hamiltonians} 
\author{Jacob Bringewatt}
\affiliation{Department of Physics, University of Maryland, College Park, Maryland 20742, USA}
\affiliation{Joint Center for Quantum Information and Computer Science, College Park, Maryland 20742, USA}
\affiliation{Joint Quantum Institute, College Park, Maryland 20742, USA}
\author{Michael Jarret}
\affiliation{Booz Allen Hamilton, Annapolis Junction, MD, 20701, USA}
\affiliation{Department of Mathematical Sciences, George Mason University, Fairfax, VA, 22030, USA}
\affiliation{Quantum Science and Engineering Center, George Mason University, Fairfax, VA, 22030, USA}
\date{\today}

\begin{abstract}
All known examples suggesting an exponential separation between classical simulation algorithms and stoquastic adiabatic quantum computing (StoqAQC) exploit symmetries that constrain adiabatic dynamics to effective, symmetric subspaces. The symmetries produce large effective eigenvalue gaps, which in turn make adiabatic computation efficient. We present a classical algorithm to sub-exponentially sample from an effective subspace of any $k$-local stoquastic Hamiltonian $H$, without \emph{a priori} knowledge of its symmetries (or near-symmetries). Our algorithm maps any $k$-local Hamiltonian to a graph $G=(V,E)$ with $\abs*{V} = \bigO{\poly(n)}$ where $n$ is the number of qubits. Given the well-known result of Babai \cite{babai2016graph}, we exploit graph isomorphism to study the automorphisms of $G$ and arrive at an algorithm quasi-polynomial in $\abs{V}$ for producing samples from effective subspace eigenstates of $H$. Our results rule out exponential separations between StoqAQC and classical computation that arise from hidden symmetries in $k$-local Hamiltonians. Our graph representation of $H$ is not limited to stoquastic Hamiltonians and may rule out corresponding obstructions in non-stoquastic cases, or be useful in studying additional properties of $k$-local Hamiltonians.
\end{abstract}

\maketitle

\paragraph{\textbf{Introduction.---}}
The power of adiabatic quantum computation (AQC) with stoquastic Hamiltonians (StoqAQC), formally introduced in \cite{bravyi2008complexity}, remains difficult to understand. While we know AQC with general Hamiltonians is universal \cite{aharonov2008adiabatic}, one might reasonably expect that stoquastic Hamiltonians -- those that have a known representation with real, non-positive off-diagonal matrix elements -- are more efficiently classically simulable. In this work, we provide a classical, quasipolynomially-efficient algorithm for sampling from eigenstates of $k$-local, stoquastic Hamiltonians which are otherwise widely believed to ``obstruct'' classical simulation algorithms \cite{Hastings,jarret2016adiabatic,bringewatt2018diffusion,hastings2020power}.

\paragraph{\textbf{Background.---}} AQC interpolates over a one-parameter family of Hamiltonians $H(s)$ to produce a quantum state close to the ground state of $H(s_f)$. The computational cost of this process is usually bounded by an adiabatic theorem scaling inversely in the minimal eigenvalue gap $\gamma_{\min} = \min_s \gamma\left(H(s)\right)$ between the two lowest eigenvalues $\gamma(H(s))=\lambda_1(H(s))-\lambda_0(H(s))$ of $H(s)$ \cite{jansen2007bounds, elgart2012note}. Thus, an efficient simulation algorithm must scale at most sub-exponentially with $\gamma_{\min}^{-1}$.

While most researchers do not expect StoqAQC to be capable of universal quantum computing (as evidenced by i.e. \cite{bravyi2017polynomial, bravyi2008complexity, bravyi2008quantum, crosson2018rapid,crosson2020classical}), the primary techniques for simulating these processes rely on Monte Carlo (MC) methods, and tend to focus on estimating certain properties of thermal states. For a full simulation of a StoqAQC process, however, we would like to be able to reproduce ground-state measurement statistics at zero temperature. When comparing AQC to particular MC-based algorithms for this task, there exist a number of ``obstructions'' that do yield exponential separations  \cite{Hastings,jarret2016adiabatic,bringewatt2018diffusion,hastings2020power}. In fact, results published shortly after the first version of this manuscript exploit these obstructions to prove a superpolynomial oracle separation between classical algorithms and StoqAQC with non-$k$-local, but sparse Hamiltonians \cite{hastings2020power}. 

Many classical methods are limited to classically sampling from a statistical distribution proportionate to probability amplitudes of a quantum state $\phi$, rather than probabilities of its measurement outcomes. One can immediately see a fundamental obstruction to this approach in its most abstract---the divergence between $\norm{\phi}_1$ and $\norm{\phi}_2$ can greatly impact sampling statistics. For our purposes, let $\phi$ be the ground state of $H(s)$ and let $\phi(i) \equiv \left\langle i \vert \phi  \right \rangle $. Now, suppose that there exists some $m$ such that $\abs*{\phi(m)}^2/\ltn{\phi}^2=\Omega(1/\poly(n))$, where $n$ is the number of qubits. An efficient quantum process capable of producing the state $\phi$ will take only $\bigO{\poly(n)}$ measurements of $\phi$ in the basis $\set*{\ket{i}}\ni \ket{m}$ to reliably return $m$. Alternatively, suppose one can produce samples of a random variable $X$ taking values in $\set*{i}$ with probability mass function (PMF) $\Pr{X=i}=\abs*{\phi(i)}/\lon{\phi}$. We have that $\frac{\abs*{\phi(m)}}{\lon{\phi}}=\frac{\abs{\phi(m)}}{\ltn{\phi}} \frac{\ltn{\phi}}{\lon{\phi}}\geq 2^{-n/2}\frac{\abs*{\phi(m)}}{\ltn{\phi}}$. When this inequality is nearly achieved, one requires exponentially many samples of $X$ before one expects to return $m$. (For an explicit example, see \cite[Example 0]{jarret2016adiabatic}.) Thus, even an efficient classical process for perfectly producing samples of $X$ may be exponentially slower than its quantum counterpart. Even in an idealized case, it seems one might require a classical process capable of directly sampling from the PMF $\Pr{X=i}=\abs*{\phi(i)}^2/\ltn{\phi}$, which can be rather difficult to obtain.

\paragraph{\textbf{Approach.---}} Current examples of $k$-local Hamiltonians where the inequality above is nearly achieved for the ground state of $H(s)$ rely on symmetries maintained by $H(s)$, constraining adiabatic dynamics to a polynomially-sized effective subspaces  \cite{farhi2002quantum, brady2016spectral, R04, CD14, MAL16, bringewatt2019polynomial}. \textbf{This raises a natural question: is it possible to efficiently classically reproduce the quantum statistics of the ground state of $H(s)$ without knowing its (near) symmetries \textit{a priori}?} In this paper, we answer in the affirmative, thereby providing an algorithm capable of simulating eigenstates previously deemed ``obstructed'' \cite{Hastings,jarret2016adiabatic,bringewatt2018diffusion,bringewatt2019polynomial}.

In particular, we consider symmetries such that the ground state $\phi$ of $H(s)$ is preserved under permutations $\pi$ of computational basis states $\set*{\ket{i}}$, such that $\abs*{\phi(m)} = \abs*{\phi(\pi(m))}$ for all $m \in \set{i}$. As elaborated later, this set of symmetries can be formally described by the automorphism group of an appropriate graph. For simplicity, in this Letter, we restrict attention to symmetries generated by terms with the same number of interacting qubits. This includes all $k$-local obstructions previously discovered, but more general constructions are possible \footnote{See \cite{bringewattjarretinprep} and the supplementary material for further discussion.}. Here, we introduce a classical algorithm that discovers and leverages such symmetries in quasi-polynomial time. The algorithm is upper bounded in its complexity by the greater of graph isomorphism (GI) on graphs with $\poly(n)$ vertices and $\poly(\abs*{S})$, where $S$ is an irreducible set of equivalence classes between computational basis states. Since GI is solvable in quasi-polynomial time \cite{babai2016graph,helfgott2017graph}, our algorithm scales quasi-polynomially in $n$ whenever $\abs*{S}$ is quasi-polynomial in $n$. This rules out exponential separations between AQC and classical algorithms for $k$-local, stoquastic Hamiltonians with large automorphism groups.

Our graph constructions are general, however our restriction to stoquastic Hamiltonians simplifies our study of symmetries. The Perron-Frobenius theorem guarantees that the ground state $\phi$ of a stoquastic Hamiltonian can be written with all non-negative amplitudes in the computational basis (e.g. \cite{horn1990}) \footnote{Actually, discrete nodal domain theorems are generally more useful here, but they too depend upon Perron-Frobenius theorem \cite{BrianDavies}.}. Automorphisms $\pi$ can be expressed as tensor products of Pauli-$X$ operators, and the non-negativity of the ground state means that we can ignore sign changes and only study the case that $\phi(m) = \phi(\pi(m))$. (Note that there exists a bitstring $b \in \set*{0,1}^n$ such that $\bigotimes_i X^{b_i} \ket{m}  = \ket{l}$ for all $l,m \in \set{0,1}^n$.) That is, the lack of sign change means that we do not need to consider the possibility of first conjugating $H(s)$ by unitaries that map $X$ terms to $Y$ terms and vice versa. Although we exploit this simplification, we anticipate these constructions can be generalized to study properties of general $k$-local Hamiltonians.

\paragraph{\textbf{Algebraic graph theory.---}}
The primary contribution of this paper is the formal construction of bijective mappings from $H$ to a pair of graphs, which allows us to reduce the problem of simulation to that of GI. The two mappings are bijective, in the sense that the Hamiltonian can be recovered from each graph. The first, $H\mapsto \hg$, takes $H$ to an exponentially-sized, undirected graph $\hg$ with spectral properties consistent with $H$ itself. The second, $H \mapsto \ctg$, maps $H$ to a vertex-colored, directed graph $\ctg=(V_\ctg,E_\ctg)$ which incorporates all relevant automorphisms of $\hg$. In a sense, the latter is the compact, graph representation of $\hg$ in the same way that $H$, written in terms of Pauli matrices, is the compact representation of $H$ as a matrix. $\ctg$ can be used to efficiently reconstruct and determine equivalent vertices of $\hg$ via GI and, thus, determine the effective subspace of $H$. We first construct $\hg$, then present the algorithm in detail in which we treat the construction of $\ctg$ as a black box. Then, we explicitly provide a construction of $\ctg$. We refer the reader to the supplemental material for a complete, minimal example.

\paragraph{\textbf{Mapping I: $H$ to $\hg$ ---}}\label{ss:hamiltonianstographs}
We consider the weighted graph representation, $\hg=\left(V_\hg,E_\hg,w_\hg\right)$, of a stoquastic Hamiltonian where $V_\hg = \set*{X_b=\bigotimes_{i} X^{b_i}}_{b \in \set*{0,1}^n} \cup \set*{\infty}$ \cite{jarret2018hamiltonian}. That is, we label each vertex usually associated with computational basis state $\ket{b}$ by $X_b$, as $X_b\ket{r} =\ket{b \oplus r}$ for any $r \in \set{0,1}^n$, and we seek an $r$-independent construction. (We also define $Y_b = \bigotimes_i Y^{b_i}$ and $Z_b = \bigotimes_i Z^{b_i}$.) We assume that we are presented with a $k$-local stoquastic Hamiltonian $H \in \mathbb{R}^{\abs*{V_\hg^*}\times \abs*{V_\hg^*}}$, where $V_\hg^* = V_\hg \setminus \{\infty\}$. Specifically,
\be\label{eqn:genH}
    H= -\sum_{\hw{b} \leq k}\alpha_{b}X_b -\sum_{\substack{\hw{b} \leq k \\ \hw{b} \in 2 \mathbb{Z}}} \beta_b Y_b + \sum_{\hw{b} \leq k}\kappa_b Z_b,
\ee
where $\hw{b}$ is the Hamming weight of the bit string $b$ and $\abs*{\beta_b} \leq \alpha_b \; \forall \; b$ \footnote{\label{fn:combgadgets} For simplicity and to avoid too much notation, we consider only Hamiltonians that can be written as these combinations. However our construction will generalize to situations with mixed terms (i.e. $XYY$), by introducing a new gadget that connects together an $X$- and $YY$-gadget.}. From $H$, we identify the set of \emph{edge generators} $K = \set*{X_b \; \vert \; {\alpha_b \neq 0} }$. Let $H_X = \sum_{\hw{b} \leq k}\alpha_{b}X_b$ and $H_Y=\sum_{\substack{\hw{b} \leq k \\ \hw{b} \in 2 \mathbb{Z}}} \beta_b Y_b$. Now,
\begin{align}
    \bra{b'}(H_X + H_Y) X_b\ket{b'} &= \alpha_b + \im^{-\hw{b}}\beta_b\bra{b'}Z_b\ket{b'} \nonumber \\
    &= \alpha_b +  \im^{-\hw{b}} (-1)^{b \cdot b'} \beta_b \nonumber \\
    &=  \alpha_b + \im^{2 b \cdot b' - \hw{b}}\beta_b.
\end{align}
We let $w(u,v)=w(v,u)$ and define edge weights,
\begin{equation}\label{eqn:weights}
    w\left(X_{b'},v\right) = 
        \begin{cases}
            \alpha_{b} + \im^{2 b \cdot b' - \hw{b}}\beta_{b} & \text{if $v=X_{b'\oplus b}$} \\
            \displaystyle\sum_{\hw{b} \leq k}(-1)^{b\cdot b'}\kappa_{b} & \text{$v=\infty$,}
        \end{cases}
\end{equation}
and edges $E_\hg = \set*{ \set*{u,v} \; \vert \; w(u,v)\neq 0}$. The eigenvectors of $H$ satisfy
\begin{equation}
    (w(u,\infty) - \lambda_i) \phi_i(u) = \sum_{v \in V_\hg^*} w(u,v)\phi_i(v),
\end{equation}
where $u\in V_\hg^*$, $\phi(\infty)=0$, and $(\phi_i,\lambda_i)$ is the $i$th eigenvector-eigenvalue pair. In order to identify symmetric subspaces of $H$, we consider identifying all vertices of $\hg$ that are equivalent under an edge-weight preserving automorphism $f:V_\hg\longrightarrow V_\hg$ of $\hg$. We call the set of all such automorphisms $\aut(\hg)$. Now, we sum over equivalence classes $\eqc{u} = \set*{f(u)}_{f \in \aut(\hg)}$: 
\[
    \sum_{u' \in \eqc{u}}(w(u',\infty) - \lambda_i)\phi_i(u') = \sum_{u' \in \eqc{u}} \sum_{v} w(u',v) \phi_i(v),
\]
or
\begin{equation}
    \left(w(u,\infty)-\lambda_i\right) \phi_i(u) = \sum_{\eqc{v}} \omega_{u \eqc{v}}\phi_i(v), \label{eqn:evec}
\end{equation}
where $\omega_{u \eqc{v}} = \sum_{v \in \eqc{v}} w(u,v)$.

This defines our effective Hamiltonian $H':\eqc{V_\hg^*}\times \eqc{V_\hg^*}\longrightarrow \mathbb{R}^+$ on the space of effective vertices $\eqc{V_\hg^*} = \set*{\eqc{u}}_{u \in V_\hg^*}$: 
\be \label{eqn:reduced_hamiltonian}
    H'(\eqc{u},\eqc{v}) = 
    \begin{cases}
        w(u,\infty) & \text{if $\eqc{u} = \eqc{v}$} \\
        - \omega_{u\eqc{v}} & \text{otherwise.}
    \end{cases}
\ee
Note that $\omega_{v\eqc{u}}\neq \omega_{u\eqc{v}}$, but rather $\abs{\eqc{v}} \omega_{v\eqc{u}} = \abs{\eqc{u}} \omega_{u \eqc{v}}$. By \cref{eqn:evec}, the right eigenvector of $H'$ corresponding to eigenvalue $\lambda_0$ is proportional to the eigenvector of $H$ corresponding to eigenvalue $\lambda_0$.

\paragraph{\textbf{The algorithm}.---}Assume that we can map our Hamiltonian to a graph $\hg$ as described above. Our goal is to find an effective graph $\hg'$ with vertex set $\eqc*{V_\hg^*}\cup \{\infty\}$, whose ground state corresponds to that of $\hg$.

For clarity, we break the classical algorithm into two parts: (1) \textproc{FindEffectiveVertices}, which recursively searches $\hg$ to return $V'$ such that $u \in V' \iff V'\cap \eqc{u}=\{u\}$; and (2) \textproc{FindEffectiveGraph}, which takes as input $V'$ and returns $\hg'$. Both routines assume the existence of an ancillary algorithm $\textproc{FindRepresentative}(u,V')=v\in V'\cap \eqc{u}$, whose existence we will later justify. For now, we treat it as an oracle with runtime quasi-polynomial in $n$, $\bigO{\abs*{V'}\quasip(n)}$, where $\quasip(n)$ matches the runtime of the best-known GI algorithm \cite{babai2016graph, helfgott2017graph}.

\begin{algorithm}[H]
 \caption{\label{alg:main} Find Effective Vertices}
 \begin{algorithmic}[1]
  \Function{FindEffectiveVertices}{$u , V'$}
    \If{$u = \emptyset$} 
    \State $u \gets \Call{Random}{V^*_\hg}$ \label{alg:main:rand}
    \State Add $u$ to $V'$
    \EndIf
    \For{$v \in N(u)$}
        \If{$\Call{FindRepresentative}{v,V'} = \emptyset$} 
        \State Add $v$ to $V'$
        \State $V' \gets \Call{FindEffectiveVertices}{v,V'}$
        \EndIf
    \EndFor
    \State \Return $V'$
 \EndFunction
 \end{algorithmic}
\end{algorithm}

\cref{alg:main} returns a set of vertices such that each vertex is distinct and the entire routine, including the \textproc{FindRepresentative} subroutine, takes time $\bigO{\Delta(\hg) \abs{V'}^2 \quasip(n)}$ where $\Delta(\hg)$ is the maximum degree of $\hg$.  Since $V'$ includes precisely one representative of each equivalence class in the connected component of $\hg$, the following routine generates the effective graph $\hg'$.

\begin{algorithm}[H]
 \caption{\label{alg:find_graph} Find Effective Graph}
 \begin{algorithmic}[1]
  \Function{FindEffectiveGraph}{$\Gamma$}
    \State $V' \gets \Call{FindEffectiveVertices}{\emptyset}$
    \State $\Omega_{u v} \gets 0$ for all $u,v \in V'$
    \For{$u \in V'$} \label{alg:find_graph:loop}
        \For{$v \in N(u)$}
            \State $v \gets \Call{FindRepresentative}{v,V'}$
            \State $\Omega_{u v} \gets \Omega_{u v} + w(u,v)$        
        \EndFor
    \EndFor 
    \State \Return $(V',\Omega)$
 \EndFunction
 \end{algorithmic}
\end{algorithm}

The primary loop of $\textproc{FindEffectiveGraph}$ (Line \ref{alg:find_graph:loop}) takes time $\bigO{\Delta(\hg) \abs{V'}^2 \quasip(n)}$, and therefore the total time to obtain the effective graph is also $\bigO{\Delta(\hg) \abs{V'}^2 \quasip(n)}$. 

We can now obtain $H'$ and sample from its eigenstates. For $u,v \in V'$, $\Omega_{uv} = \omega_{u \eqc{v}} = \sum_{v_0 \in \eqc{v}} w(u,v_0)$. Thus, \cref{eqn:reduced_hamiltonian} is well-defined and the operator $H'$ known, even if each entire equivalence class $\eqc{u}$ is not. 

We know that existing methods, such as the power iteration method, can produce the ground state $\phi'$ of $H'$ with error $\epsilon$ in time $\bigO{\log(\epsilon^{-1})/\log(\lambda_1/\lambda_0)}$. Therefore, we can sample the ground state of the full Hamiltonian $H$ in time $\bigO{\log(\epsilon^{-1})/\log(\lambda_1/\lambda_0) + \Delta(\hg)\abs*{V'}^2\quasip(n)}$.

We note that we cannot simply normalize $\phi'$ and expect to obtain appropriate statistics; rather, each $u \in V' \cap \eqc{u}$ represents a sample of the class itself. Thus, we need to sample $\eqc{u}$ with probability $\abs*{\eqc{u}}\phi(u)^2$, where $\phi$ is the appropriately normalized ground state of $H$. By \cref{eqn:evec}, $H'$ has a ground state $\phi'$ that preserves relative amplitudes $\frac{\phi(u)}{\phi(v)} = \frac{\phi'(u)}{\phi'(v)}$ for all $u,v \in V'$. 

Now, we use $\phi'$ and $\abs*{\eqc{u}}$ to sample $u \in \eqc{u} \cap V'$ with probabilities according to $\phi$, 
\begin{equation}\label{eqn:prob}
    {\Pr{\eqc{u}} = \abs*{\eqc{u}}\phi(u)^2 = \frac{\abs*{\eqc{u}}\phi'(\eqc{u})^2}{\sum_{v \in V'} \abs*{\eqc{v}}\phi'(\eqc{v})^2}.}
\end{equation}
Note that for $\omega_{v\eqc{u}}\neq 0$, $\frac{\abs{\eqc{v}}}{\abs{\eqc{u}}}=\frac{\omega_{u\eqc{v}}}{\omega_{v\eqc{u}}}$. Therefore, 
\begin{align}
    \frac{\abs*{\eqc{u}}}{\sum_{v \in V'}\abs{\eqc{v}}} =\left(\sum_{v \in V'} \prod_{e\in P(u,v)} \frac{\omega_{e_0\eqc{e_1}}}{\omega_{e_1\eqc{e_0}}}\right)^{-1},
    \label{eqn:numclasses}
\end{align}
where $P(u,v) \subseteq E_{G'}$ is any directed path connecting $u,v \in V'$. Up to a factor constant for all $u,v$, \cref{eqn:numclasses} determines $\abs*{\eqc{u}}$ and, thus, fully determines \cref{eqn:prob}. 

Repeating this process initializes a new seed in \cref{alg:main} Line \ref{alg:main:rand}, and we return each member of $\eqc{u}$ with equal probability. Furthermore, the random seed guarantees that a sample from a connected set of vertices $V_C$ of $\Gamma$ is returned with probability $\abs*{V_C}/\abs*{V_\Gamma}$, as expected.

\paragraph{\textbf{Mapping II: $H$ to $\ctg$---}} Now, we explicitly give an implementation of $\textproc{FindRepresentative}$. It is helpful to keep in mind that while this construction is unavoidably definition-heavy, the construction naturally reduces our problem to GI. In fact, our approach is somewhat similar to Luks' reduction of graph automorphism to GI \cite{luks1993} or Crawford's formalism of symmetries in clausal theories \cite{crawford1992theoretical}, applied to the study of $\aut(\hg)$. Inspired by the latter, we build what we abusively call a \emph{clausal theory graph} $\ctg$ and our goal is to define an invertible map $M$ such that $M(\hg) = M_0[V_\hg]\cup M_1[E_\hg] = \ctg$. We do so by introducing \emph{gadgets}, smaller graphs that allow us to separately map each $v\in V_\hg$ and $e \in E_\hg$ to specific vertex-colored, directed graphs.  The union of these gadgets forms $\ctg$. We will introduce a number of different types of vertices, where each type is assigned a unique color represented by a superscript (E.g. $\ell^{(\mathpzc{a})}$). In the following $\set{\mathpzc{a}, \mathpzc{b}, \mathpzc{c}, \mathpzc{d}}$ represent fixed, unique colors and $\set{\mathpzc{x}_b, \mathpzc{y}_b, \mathpzc{z}_b}$ represent a distinct set of variable colors which are assigned based on the particular Hamiltonian. Two vertices are identical only if they both have the same color and label. Furthermore, for simplicity, we will abusively write $\{u,v\}$ for an undirected edge and $(u,v)$ for a directed edge. (Thus, $\{u,v\} \in E$ can be read as $\set*{(u,v),(v,u)}\subset E$.) 

First, we define a set of \emph{literals} $L = \set*{Z_{\vec{i}}^{(\mathpzc{a})}}_{i=0}^{n-1}$ and their negations $- L = \set*{-Z_{\vec{i}}^{(\mathpzc{a})}}_{i=0}^{n-1}$, where $\vec{i} = (\delta_{ij})_{j=0}^{n-1}$. We label each vertex $X_b \in V_\hg^*$ by a set of literals  $A(X_b)^{(\mathpzc{b})} = \set*{(-1)^{b_i}Z_{\vec{i}}^{(\mathpzc{a})}}_{i=0}^{n-1}$. We call $A(X_b)^{(\mathpzc{b})}$ an \emph{assignment}. Each $X_b$ corresponds to a gadget, the vertex-colored star graph $M_0(X_b)$ with edge set $E_{M_0(X_b)}=\set*{\set*{\ell^{(\mathpzc{a})},A(X_b)^{(\mathpzc{b})}}}_{\ell^{(\mathpzc{a})} \in A(X_b)^{(\mathpzc{b})}}$. Furthermore, $M_0:V_\hg\longrightarrow \left[L\cup -L\right] \cup \set*{A(X_b))^{(\mathpzc{b})}}_{b \in \set*{0,1}^n}$ is bijective and hence invertible. Thus, $M_0^{-1}(M_0(X_b)) = X_b$.

For each edge generator $X_b \in K$, we construct the graph $G_1(b)$ specified by edge set $E_{G_1(b)}=\bigcup_{b_i \neq 0} \set*{\{Z_{\vec{i}}^{(\mathpzc{a})},X_b^{(\mathpzc{x}_b)}\},\{X_b^{(\mathpzc{x}_b)},-Z_{\vec{i}}^{(\mathpzc{a})}\}}$. Here, $\mathpzc{x}_b = \alpha_b$ (i.e. we assign these vertices a color based on the corresponding coefficient in the Hamiltonian) and we name such vertices \emph{generator vertices}.

Each $G_1(b)$ only captures weights $\alpha_b$ corresponding to edges generated by $X_b\in K$. We still require gadgets that incorporate $\beta_b$, so that we can extract edge weights consistent with \cref{eqn:weights} from $\ctg$. Define
\begin{equation}\label{eq:ub}
U_b =
        \set*{\set*{(-1)^{b_i'}Z_{\vec{i}}^{(\mathpzc{a})}}_{b_i \neq 0}^{(\mathpzc{c})} \vert \im^{2b\cdot b'-\hw{b}}\beta_b<0}_{b'\in\{0,1\}^{n}}.
\end{equation}
Note that when $\beta_b = 0$, $U_b = \emptyset$.

To specify the gadget, we construct the directed vertex-colored graph $G_2(b)=\bigcup_{u_b^{(\mathpzc{c})} \in U_b} g(u_b^{(\mathpzc{c})})$, where each $g(u_b^{(\mathpzc{c})})$ is the star graph with edge set $E_{g(u_b^{(\mathpzc{c})})}= \set*{(\ell^{(\mathpzc{a})},u_b^{(\mathpzc{c})}),(u_b^{(\mathpzc{c})},- \ell^{(\mathpzc{a})}),\{u_b^{(\mathpzc{c})},Y_b^{(\mathpzc{y}_b)}\}}_{\ell^{(\mathpzc{a})} \in u_b^{(\mathpzc{c})}}$. We name the $u_b^{(\mathpzc{c})}$s and the $Y_b^{(\mathpzc{y}_b)}$s \emph{weight generator} and \emph{weight generator cluster} vertices, respectively. Here, $\mathpzc{y}_b = \max_{b'}\alpha_{b'}+\abs*{\beta_b}$ is the color representing the cluster that allows us to extract edge weights of $\hg$.

Finally, we build a graph from the term $\sum_{\hw{b} \leq k}\kappa_b Z_b$, where it helps to write $\kappa_b Z_b = \abs*{\kappa_b}C_b$.
We call each $C_b$ a \emph{clause}, and we identify the set of assignments that ``satisfy'' the clause. For a choice of $b$,
\begin{equation}\label{eq:cb}
    \mathscr{C}_{b} = \set*{\set*{(-1)^{b'_i}Z_{\vec{i}}^{(\mathpzc{a})}}_{b_i\neq 0}^{(\mathpzc{d})} \vert (-1)^{b\cdot b'} = \sign(\kappa_b)}_{b'\in\{0,1\}^{n}}.
\end{equation} 
As $H$ is $k$-local, $\abs*{\mathscr{C}_b}= 2^{\abs{b}-1}\leq 2^{k-1}$. Now, we construct the edge set $E_{G_3(b)} = \set*{\set*{c^{(\mathpzc{d})},Z_b^{(\mathpzc{z}_b)}}\cup\set*{c^{(\mathpzc{d})},\ell^{(\mathpzc{a})}}_{\ell^{(\mathpzc{a})}\in c^{(\mathpzc{d})}}}_{c^{(\mathpzc{d})} \in \mathscr{C}_b}$. We name $c^{(\mathpzc{d})}$'s and $Z_b^{(\mathpzc{z}_b)}$'s \emph{clause} and \emph{clause cluster} vertices, respectively. Here, $\mathpzc{z}_b = \max_{b'} \alpha_{b'} + \max_{b'}\abs{\beta_{b'}}+\abs{\kappa_b}$ is the color representing the cluster of satisfying assignments, which allows us to extract edge weights of $\hg$. These gadgets are summarized in \cref{tab:gadgets} \footnote{As previously mentioned, for full generality one must consider composite gadgets for Hamiltonians with mixed terms (i.e. XYY or XYZ). See supplementary material for details.}. 

\begin{table}
    \centering
    \includegraphics[width=\columnwidth]{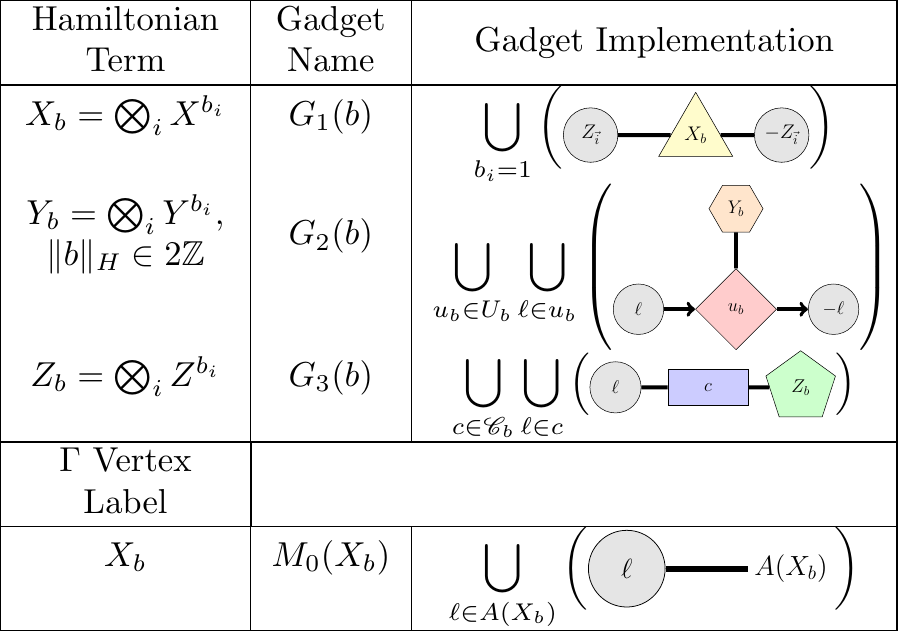}
    \caption{Gadgets summary. Vertex labels are as defined in the text (See \cref{eq:ub,eq:cb}). Superscripts are dropped and represented with a unique color/shape. }
    \label{tab:gadgets}
\end{table}
 
Given these constructions we define the direct mapping $\ctg=M(\hg)=M_0[V_\hg]\cup M_1[E_\hg]$ as follows:
\begin{enumerate}
    \item $M_0(u \in V_\hg)$ as defined above,
    \item $M_1(\{u,X_b u\} \in E_\hg) = G_1(b)\cup G_2(b)$, and
    \item $M_1(\{u, \infty\}\in E_\hg) = \bigcup_{\kappa_b \neq 0} G_3(b)$.
\end{enumerate}
Note that $M_1[E_\hg]$ contains all relevant information about $\aut(\Gamma)$ and, thus, we can study $\ctg' = M_1[E_\hg] \cup \left(L \cup \neg L, \emptyset \right)$. Note that when $\hg$ is connected, $\ctg'=M_1[E_\hg]$. Importantly, by construction, $\abs*{V_{\ctg'}}=\bigO{\poly(n)}$. $M$ also has the following useful property. 

\begin{thm}\label{thm:M_inverse}
    The function $M:\hg \mapsto \ctg$ is bijective. 
\end{thm}
\cref{thm:M_inverse} is true by construction, but we include proof in the supplemental material for completeness. Now, we define $\ctg(u) = M_0(u) \cup M_1[E_\hg]$ and can state the following theorem.

\begin{thm} \label{thm:automorph}
    There exists a color-preserving isomorphism $\ctg(u) \simeq \ctg(v)$ if and only if $u \equiv v$. 
\end{thm}
\cref{thm:automorph} is also true by construction, and explicit proof can be found in the supplemental material. By exploiting the $k$-local structure of the Hamiltonian in the form of a compact $\ctg$, we are able to reduce our problem from deciding whether $f\in \aut{\hg}$ for an exponentially-sized graph $\hg$ to deciding isomorphism $\ctg(u) \simeq \ctg(f(u))$ of polynomially-sized graphs $\ctg(u),\ctg(f(u))$. 

Armed with this construction and the above theorems we can give the algorithm for \textproc{FindRepresentative}.

\begin{algorithm}[H]
 \caption{\label{alg:find_representative} Check equivalent vertices}
 \begin{algorithmic}[1]
  \Function{FindRepresentative}{$u,V'$}
  \State $\mathbb{G} \gets G(u)$
    \For{$v \in V'$}
        \State $\mathbb{G}' \gets G(v)$
        \If{$\mathbb{G} \simeq \mathbb{G}'$} \Return $v$  \EndIf \label{alg:find_representative:iso}
    \EndFor
    \State \Return $\emptyset$
 \EndFunction
 \end{algorithmic}
\end{algorithm}
We note that $\ctg(u)$ can be constructed in time $\bigO{\poly(n)}$ and that color-preserving GI on bipartite, directed graphs is GI-complete \cite{zemlyachenko1985graph}. Additionally, \cref{alg:find_representative} with stoquastic $k$-local Hamiltonians is GI-complete. To see this, for any two graphs $S,S'$, label vertices such that $V_{S}\cap V_{S'} = \emptyset$ and let $H=\sum_{\set*{i,j} \in E_{S}}Z_{\vec{i}}Z_{\vec{j}} + \sum_{\set*{i,j} \in E_{S'}}Z_{\vec{i}}Z_{\vec{j}}$. Then, $\ctg\left(X_{\bigoplus_{i \in V_S} \vec{i}}\right) \simeq \ctg\left(X_{\bigoplus_{i\in V_{S'}}\vec{i}}\right)$ iff $S \simeq S'$. In the other direction, \cref{alg:find_representative} uses GI as a subroutine. The best known algorithm for GI takes time $\quasip(n) = 2^{\bigO{\log(n)^{\bigO{1}}}}$ \cite{babai2016graph,helfgott2017graph}, and therefore the entire routine takes $\bigO{\abs*{V'}\quasip(n)}$. 

\paragraph{\textbf{Discussion.---}}

Our results can be extended to near-symmetries via a straightforward application of the Davis-Kahan $\sin \Theta $ theorem \cite{yu2015useful, Bhatia}. In particular, let $H,\Delta$ be Hamiltonians where $H$ has ground state density matrix $\rho$ and $H+\Delta$ has ground state density matrix $\rho_\Delta$ from which we would like to sample. Then,
\[
{\sqrt{1-F(\rho,\rho_\Delta)} \leq \frac{\pi}{2}\frac{ \norm{\left(I-\rho_\Delta\right)\Delta\rho}_F}{\lambda_1(H) - \lambda_0(H+\Delta) }
\leq\frac{\pi}{2}\tan^2\Phi}
\]
where $F(\rho, \rho_\Delta)=\|\rho \rho_\Delta\|_F^2$ is the fidelity, $\tan^2\Phi=\frac{\abs*{\langle \Delta \rangle_{\rho}}}{\gamma(H) - \langle \Delta \rangle_{\rho}}$, and $\gamma(H)$ is the eigenvalue gap of $H$. Thus, $F(\rho,\rho_\Delta) \geq 1 - \frac{\pi^2}{4}\tan^4\Phi$.

If one has \textit{any} procedure for producing a guess $\rho$, one can later check that $\langle \Delta \rangle_{\rho}$ is small enough to guarantee $\norm*{\rho - \rho_\Delta}\leq \epsilon$. As a limited example, suppose one perturbs each $\alpha_b,\beta_b,\kappa_b$ in $H$ by at most $\delta$. Then, $\norm{\Delta}_F \leq \delta \norm{H}_F$. Therefore, provided that $\delta \leq \epsilon \frac{\gamma}{\norm{H}_F}$, $\tan^2\Phi \leq \frac{\epsilon}{1-\epsilon}$. Hence, we can achieve arbitrary precision $\epsilon$ while perturbing each of $\alpha_b,\beta_b,\kappa_b$ by $\delta = \bigO{\epsilon\gamma/\norm{H}_F}$, where we have assumed $\gamma/\norm{H}_F = \Omega\left( \poly^{-1}n\right)$ throughout. 

Alternatively, an approximate GI algorithm \cite{arvind2012approximate} might suffice to implement \textproc{FindRepresentative}; more general approximation algorithms are left for future work.

\textbf{Conclusion.---} Our algorithm rules out the existence of an exponential separation between classical algorithms and StoqAQC using Hamiltonians with effective subspaces with size $\abs*{V'}$ scaling subexponentially in $n$, a class containing all previously known $k$-local obstructions \cite{Hastings, jarret2016adiabatic, bringewatt2018diffusion}.

Beyond these symmetric and approximately symmetric problems, whether all $k$-local stoquastic Hamiltonians are quasi-polynomially simulable remains an open question. We conjecture that families of Hamiltonians that lack near-symmetries typically have exponentially small gaps, suggesting that they are difficult for AQC \footnote{A large spectral gap is not \textit{necessary} for successful AQC. It is not too difficult to construct examples where StoqAQC succeeds despite a small gap, though they often appear pathological and one might expect classical methods to be similarly successful.}. This conjecture is largely motivated by the fact that avoiding exponentially small gaps requires pathologically smooth transitions, as explained in the supplemental material \footnote{As noted in the supplement, this result is similar to, but in terms of gap-analysis, stronger than that in \cite{farhi2008make}}.  Proving this, combined with our results here and a better understanding of those near-symmetries that we can efficiently approximate, would reduce understanding the simulability of StoqAQC to better understanding the significance of the gap in both classical and quantum cases.

\textbf{Acknowledgements.---} The authors would like to thank Andrew Glaudell and Brad Lackey for useful discussions and suggestions. J.B. was supported in part by the DOE CSGF program (award No. DE-SC0019323). M.J. was supported in part by AFRL (award No. FA8750-19-C-0044). \PIRA  \, J.B. also acknowledges partial funding by the DOE BES Materials and Chemical Sciences Research for Quantum Information Science program (award No. DE-SC0019449), DOE ASCR FAR-QC (award No. DE-SC0020312), NSF PFCQC program, DOE ASCR Quantum Testbed Pathfinder program (award No. DE-SC0019040), AFOSR, ARO MURI, ARL CDQI, and NSF PFC at JQI.

\newpage
\onecolumngrid
\begin{center}
\textbf{Supplemental material for ``Effective gaps are not effective: quasipolynomial classical simulation of obstructed stoquastic Hamiltonians''} 
\end{center}
\twocolumngrid

\section{S1: Minimal Example}
\label{app:minexample}
Here we give a minimal, working example of our algorithm. We consider a Hamming-symmetric Hamiltonian 
\begin{equation}
    H=-\sum_{\hw{b}=1}X_b + \sum_{\hw{b}=1} Z_b,
    \label{eq:hamsym}
\end{equation}
and define the Hamming-symmetric Hamiltonian with base $v$ as $H_{v}=X_{v} H X_{v}$, yielding
\begin{equation}
    H_v=-\sum_{\hw{b}=1}X_b +\sum_{\hw{b}=1} (-1)^{v\cdot b} Z_b.
    \label{eqn:Hexample}
\end{equation}

This corresponds to a graph $\hg$ where $V^*_\hg$ is a hypercube with all edge weights $\alpha_b=1$. The vertex $\set{\infty}$ is connected to every vertex $u \in V^*_\hg$ by an edge weight $w(u,\infty) = \sum_{\abs{b}=1} (-1)^{v \cdot b + u \cdot b}$. In practice, this graph $\hg$ is too large to handle, but for illustrative purposes, we consider only 3 qubits so that we can track the whole algorithm by hand. Furthermore, in order to be specific, we assume that $v=010$. That is we obtain,

\begin{align*}
H_{010}&= -X_{100}-X_{010}-X_{001}+Z_{100}-Z_{010}+Z_{001} \\
&=-X_{\vec{0}}-X_{\vec{1}}-X_{\vec{2}}+Z_{\vec{0}}-Z_{\vec{1}}+Z_{\vec{2}}
\label{eqn:Hexample}
\end{align*}
which has a corresponding graph $\hg$ depicted in \cref{fig:rampexamplehg}.

\begin{figure}[h]
\includegraphics[width=70mm]{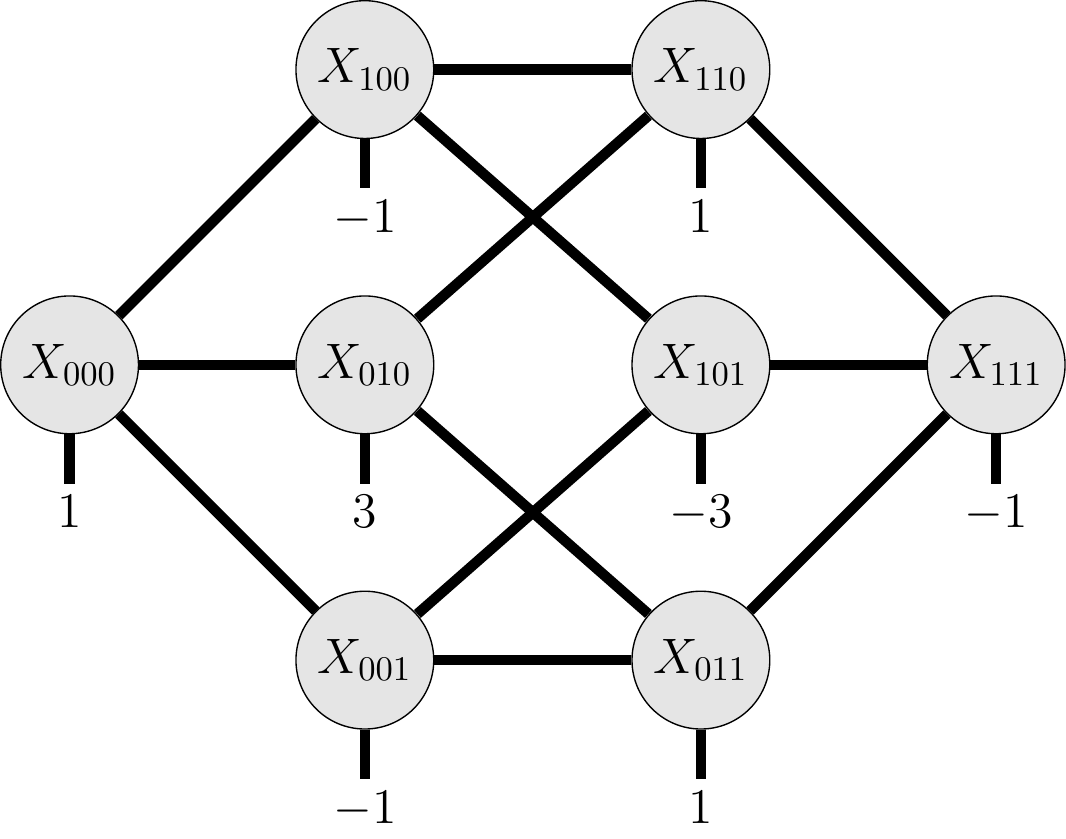}
\caption{$\hg$ for $H_{010}$. Each vertex is labeled by a member of $X_b$. The disconnected edges connect to the boundary vertex $\infty$ and are labeled by their weights.}
\label{fig:rampexamplehg}
\end{figure}

Below, we follow the steps in the main paper to obtain the clausal theory graph for $H_{010}$ excluding assignments, $M_1(E_\ctg)$ (See \cref{fig:rampexamplectg}). 

In particular, as we have three qubits we introduce the six literal vertices $Z_{\vec{0}}^{\mathpzc{(a)}}, -Z_{\vec{0}}^{\mathpzc{(a)}}, Z_{\vec{1}}^{\mathpzc{(a)}}, -Z_{\vec{1}}^{\mathpzc{(a)}}, Z_{\vec{2}}^{\mathpzc{(a)}}$ and  $-Z_{\vec{2}}^{\mathpzc{(a)}}$. We also have three single bit flip edge generators $X_{\vec{0}}, X_{\vec{1}}$ and  $X_{\vec{2}}$. Each generator has an associated vertex in the clausal theory graph, which joins its respective literals by an undirected edge. That is, since $X_{\vec{0}}$ maps the computational basis state $00\dots0$ to $10\dots0$, we use a gadget to connect $Z_{\vec{0}}$ to $Z_{\vec{1}}$. 

As there are no Pauli-$Y$ terms in the Hamiltonian, we introduce no weight generators. Finally, we add clause vertices corresponding to the diagonal potential terms of the Hamiltonian. In particular, from \cref{eqn:Hexample}, we introduce three clause cluster vertices $Z_{\vec{0}}^{(\mathpzc{z}_b=1)}, -Z_{\vec{1}}^{(\mathpzc{z}_b=1)}$ and  $Z_{\vec{2}}^{(\mathpzc{z}_b=1)}$ corresponding to to their respective terms in the Hamiltonian. For each clause cluster vertex, we identify the associated set of clause vertices. Recall for a particular choice of $b\in\{100, 010, 001\}$, we have the associated clause cluster set

\begin{equation}
    \mathscr{C}_{b} = \set*{\set*{(-1)^{b'_i}Z_{\vec{i}}}_{b_i\neq 0}^{(\mathpzc{d})} \vert (-1)^{b\cdot b'} = \sign(\kappa_b)}_{b'\in\{0,1\}^{n}}.
\end{equation} 
where in our case $\kappa_{100}=\kappa_{001}=1$ and $\kappa_{010}=-1$. Consider, for example, the $\mathscr{C}_{b = 100}$ case corresponding to the clause cluster vertex $Z_{\vec{0}}^{(\mathpzc{z}_b=1)}$. Then, $(-1)^{b\cdot b'} = \sign(\kappa_b) = 1$, whenever $b'\in\{000, 001, 010, 011\}$. Hence, noting that $b_i \neq 0$ only when $i=0$, $\mathscr{C}_{100}= \set*{\set*{Z_{\vec{0}}}^{(\mathpzc{d})}}$. Similarly, the clause cluster vertices $-Z_{\vec{1}}^{(\mathpzc{z}_b=1)}$ and $Z_{\vec{2}}^{(\mathpzc{z}_b=1)}$ have single associated clause vertices $\{-Z_{\vec{1}}\}^{(\mathpzc{d})}$ and $\{Z_{\vec{2}}\}^{(\mathpzc{d})}$ respectively. The clause cluster vertices connect their associated clause vertex to its corresponding literal vertices or, in this case, vertex. 

\cref{fig:rampexamplectg} shows the final construction. One can see that including assignments for $u,v\in\hg$ demonstrates $u\equiv v \iff \ctg(u)\simeq\ctg(v)$. This will be shown explicitly in the following walk-through of the algorithm.

\begin{figure}[t]
\includegraphics[width=70mm]{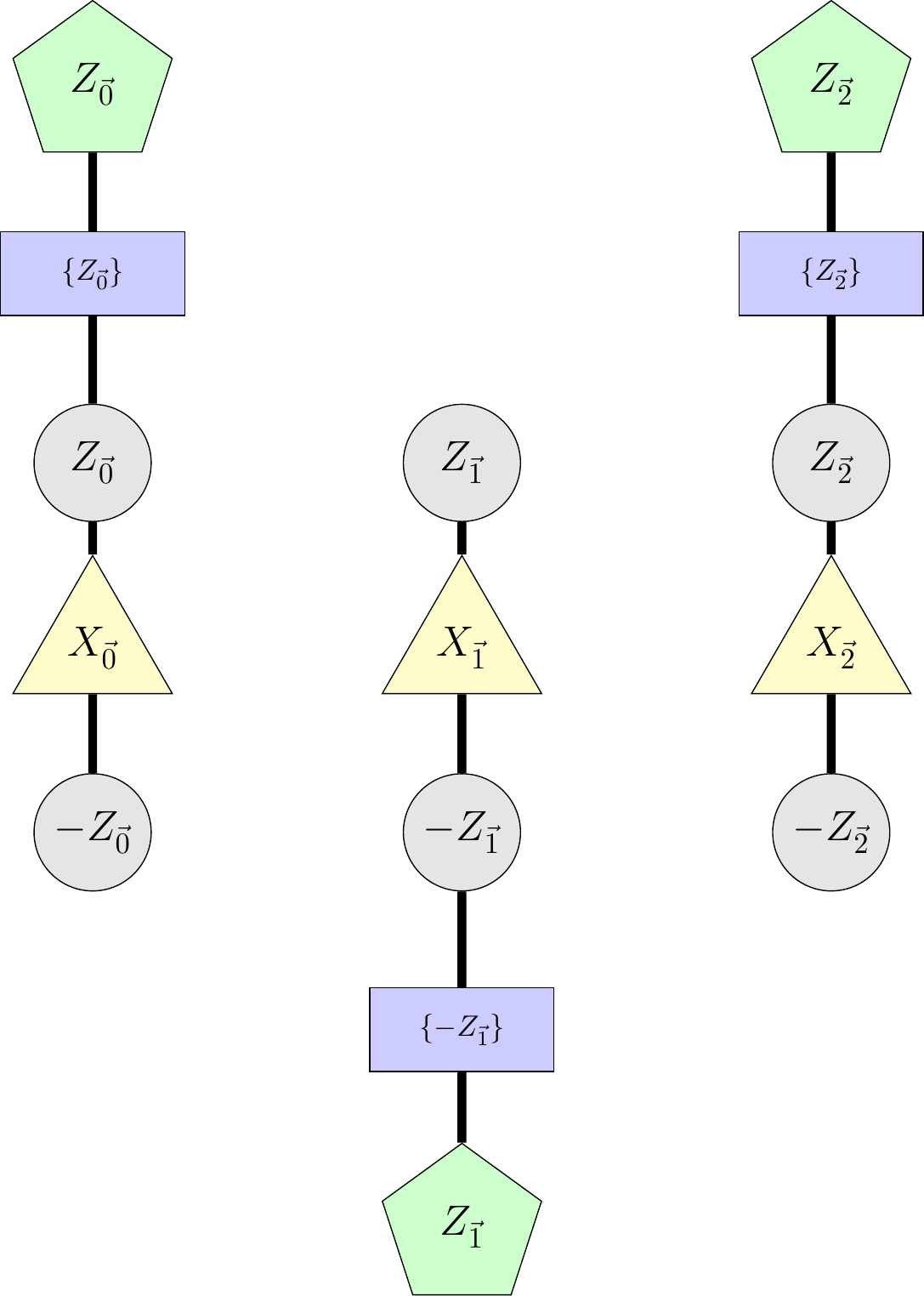}
    \caption{Full example of $M_1(E_\hg)$ for $H_{010}$.}
    \label{fig:rampexamplectg}
\end{figure}

 We start by calling Algorithm 1 with a null argument, or $\textproc{FindEffectiveVertices}(\emptyset)$. Then, in Line 3, we choose a random vertex $u \in V_\hg^*$. Suppose that $u=X_{100}$; then Line 4 adds it to $V'$, so that now $V' = \set*{X_{100}}$. We now reach Line 5 with $N(X_{100}) ={X_{000},X_{110},X_{101}}$. The loop proceeds to recursively check each neighbor $v \in N(X_{100})$. Line 6 calls Algorithm 3 with argument $(v,V')$ and if $v$ is not equivalent to an element already in $V'$, Line 7 adds $v$ to $V'$. Then, in Line 8, the Algorithm 1 calls itself with argument $(v,V')$. 
 
 To walk through these steps, suppose that in the first iteration of the loop, the first neighbor we query is $v = X_{000}$. We generate $\ctg(X_{100})$ and $\ctg(X_{000})$ as shown in \cref{fig:comp1a} and \cref{fig:comp1b}, respectively. One can see that these graphs are not isomorphic, so we add $X_{000}$ to $V'$. At this point, $V'=\set{X_{100}, X_{000}}$.

Having added $X_{000}$ to $V'$, we continue the recursion and call $\textproc{FindEffectiveVertices}(X_{000},V')$. Since $v\neq \emptyset$, we immediately proceed to the loop in Line 5 and begin checking the neighbors of $X_{000}$. Suppose the first neighbor called is $X_{001}$. (See \cref{fig:comp1c} for the corresponding clausal theory graph.) Now, we see that $\ctg(X_{001})\simeq\ctg(X_{100})$, so that $X_{001}\equiv X_{100}$. Because $X_{100} \in V'$, we do not add $X_{001}$ to $V'$, this branch of the recursion terminates, and we continue checking the remaining neighbors of $X_{000}$.

Next, we check $X_{010}$, find that it is not equivalent to any member of $V'$, and add to $V'$. Now, we call $\textproc{FindEffectiveVertices}(X_{010},V')$. Note that in this call, all of the members of $N(X_{010})$ are already equivalent to members of $V'$. Thus, Lines 5-8 complete without adding any new members to $V'$, and we return $V'$ unchanged. We now return to the parent process, which had been considering members of $N(X_{000})$. Since each member of $N(X_{000})$ has been queried, Line 9 returns $V'$ to its parent process, which was checking the members of $N(X_{100})$. We next query $X_{110}$ and find $X_{110}\equiv X_{000}$, so $X_{110}$ is not added to $V'$. In the next iteration of the loop, however, $X_{101}$ is not equivalent to anything already in $V'$ so it is added and we call $\textproc{FindEffectiveVertices}(X_{101},V')$.

The remaining unchecked neighbor of $X_{101}$ is $X_{111}$, which is equivalent to $X_{100}$, so we return $V'$ to its parent process and, in turn, the algorithm completes and exits, returning $V'=\set{X_{100}, X_{000}, X_{010}, X_{101}}$. Thus, as expected, we have one representative of each equivalence class of vertices of the original graph $\hg$ in \cref{fig:rampexamplehg}.

We note that in this example every vertex is checked, so it may not be immediately obvious that this procedure is efficient in general. However, as described in the main text, this algorithm requires checking only a polynomial number of vertices (provided the symmetries lead to a polynomial sized effective subspace). See \cref{fig:scaling} to see which vertices need to be checked by this algorithm for a Hamming symmetric example $H_v = H_{\set*{00\dots0}}$ with $5,6,$ and $7$ qubits. It is clear that while the number of vertices in the graph increases exponentially with the number of qubits, only polynomially many vertices are actually queried. 

\begin{figure*}[t]
\qquad\subfloat[$\ctg(X_{100})$]{
\includegraphics[width=50mm]{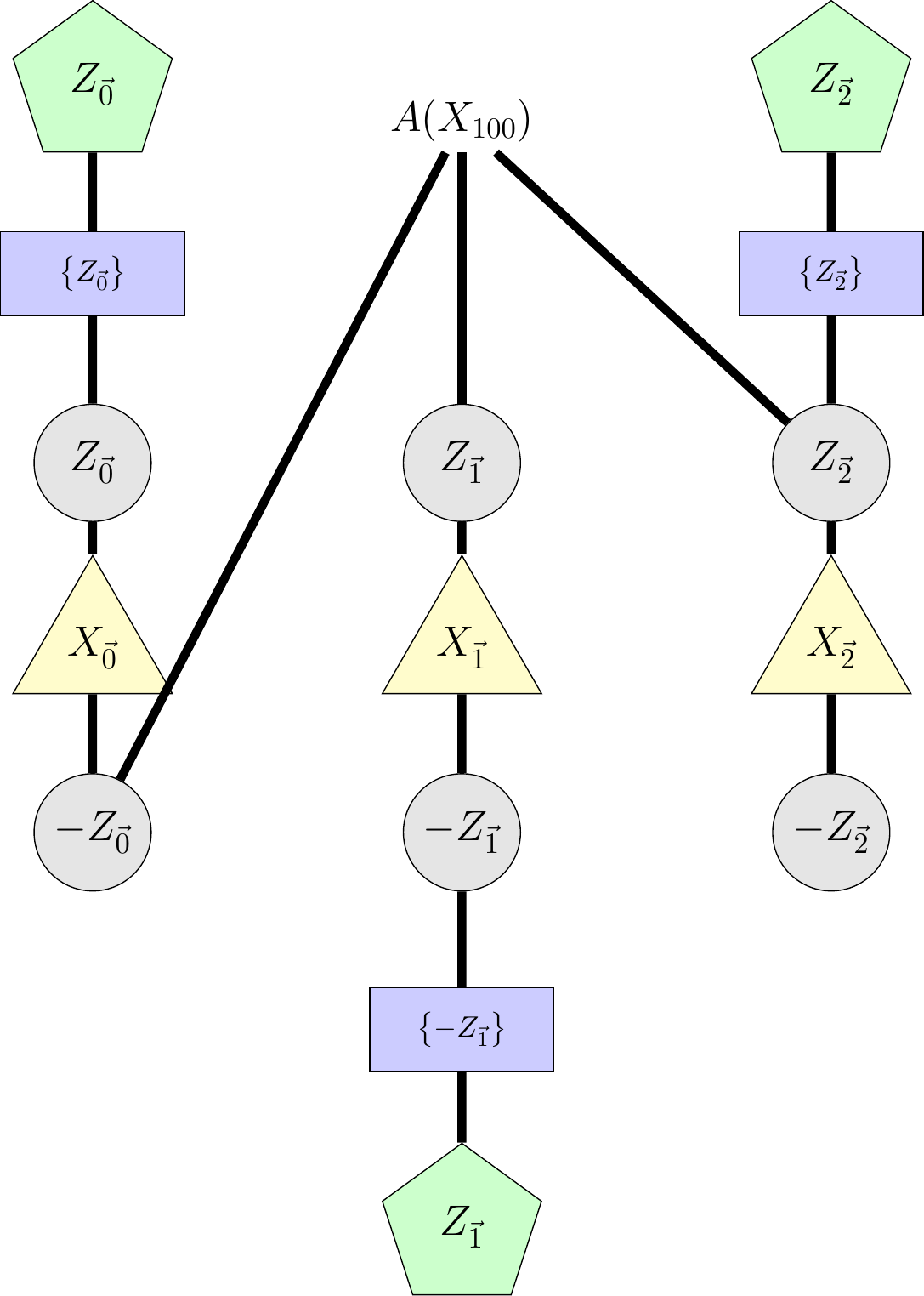}
\label{fig:comp1a}
}
\qquad\subfloat[$\ctg(X_{000})$]{
\includegraphics[width=50mm]{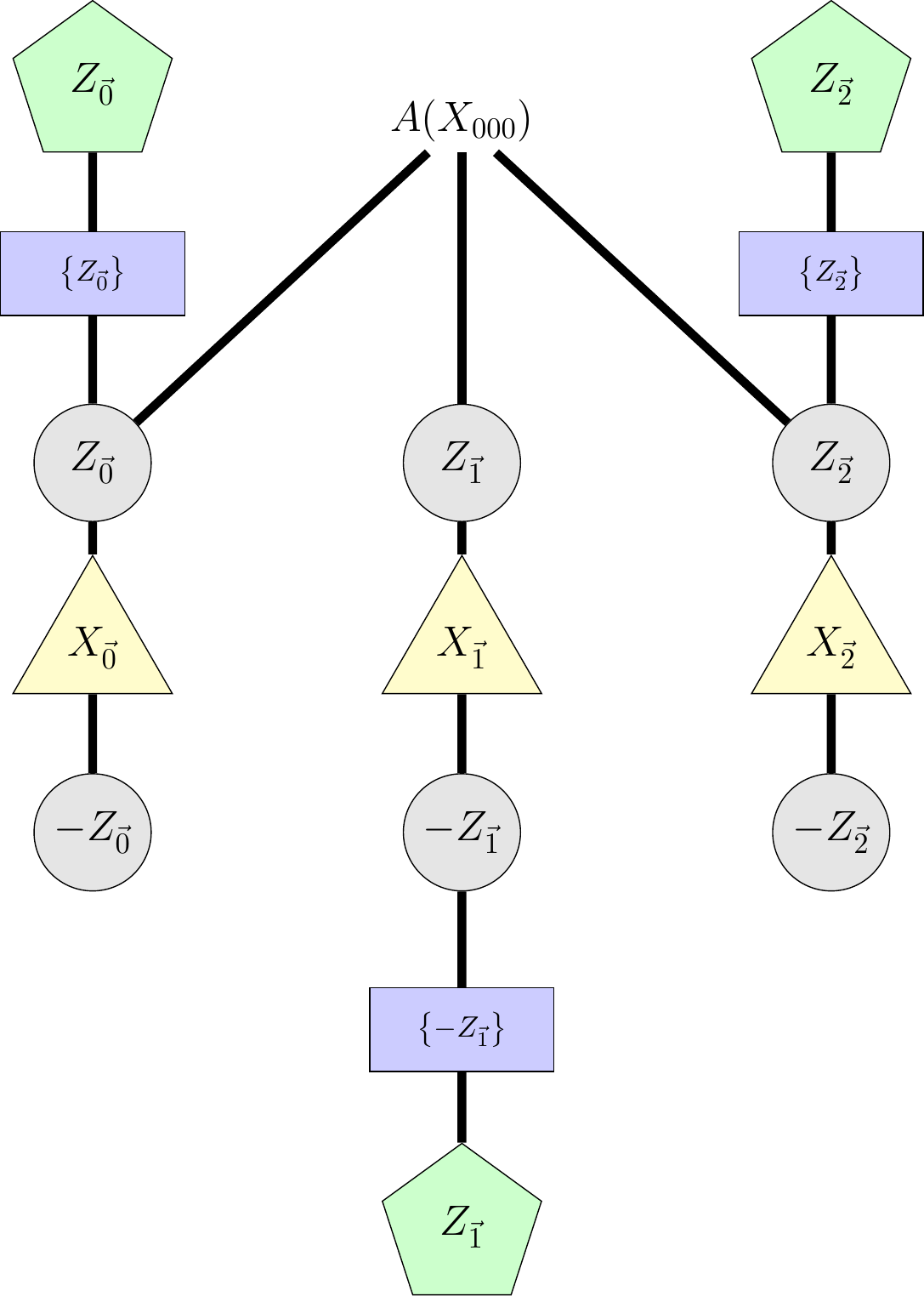}
\label{fig:comp1b}
}
\qquad\subfloat[$\ctg(X_{001})$]{
\includegraphics[width=50mm]{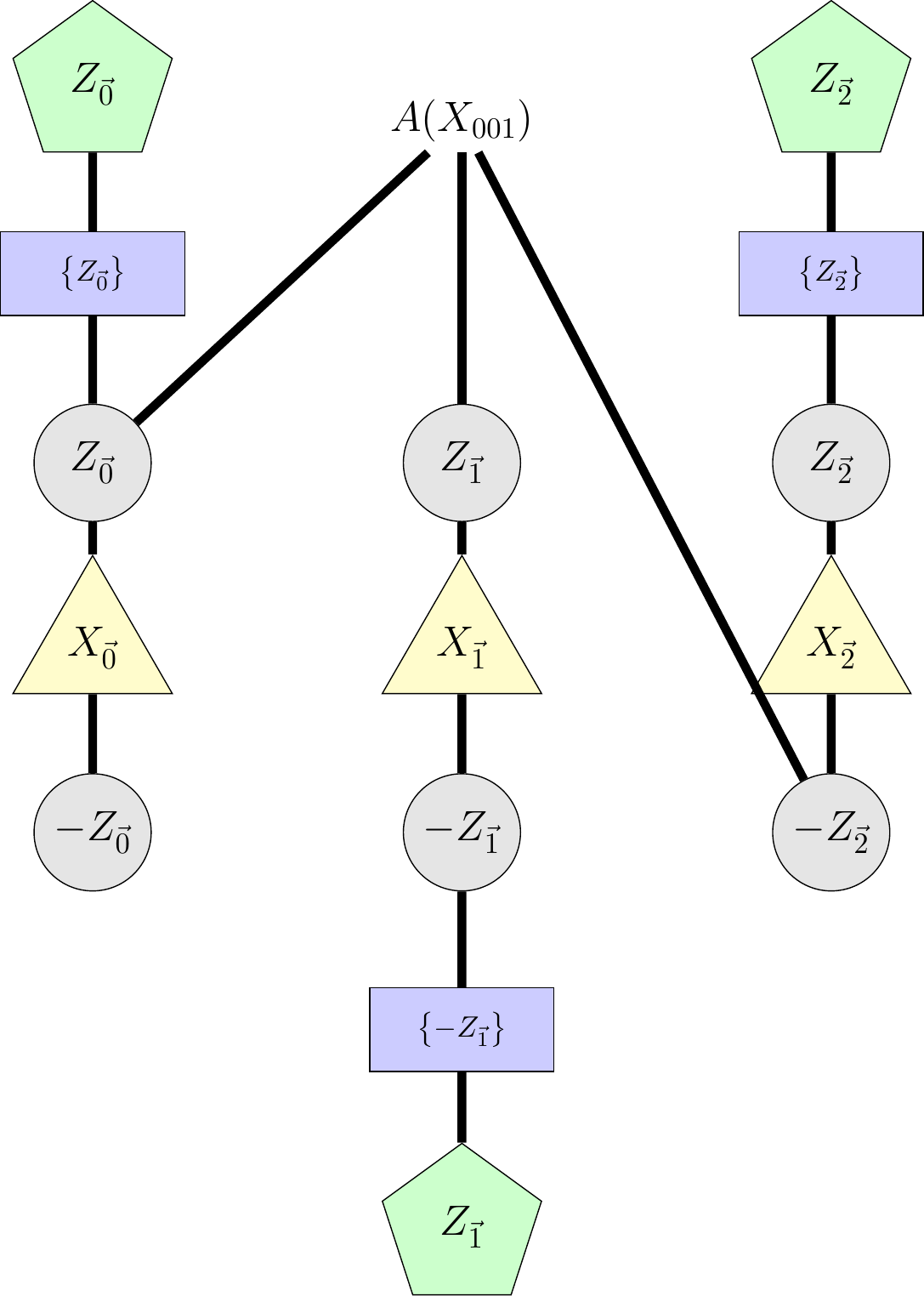}
\label{fig:comp1c}
}
\caption{Some clausal theory graphs used in the minimal example. Comparing \cref{fig:comp1a} to \cref{fig:comp1b} shows $\ctg(X_{100})\not\simeq\ctg(X_{000}) \iff X_{100}\not\equiv X_{000}$. Comparing \cref{fig:comp1a} to \cref{fig:comp1c} shows  $\ctg(X_{100})\simeq\ctg(X_{001}) \iff X_{100}\equiv X_{001}$.}
\end{figure*}

\begin{figure*}[t]
    \centering
    \qquad\subfloat[$n=5$]{
    \includegraphics[width=50mm]{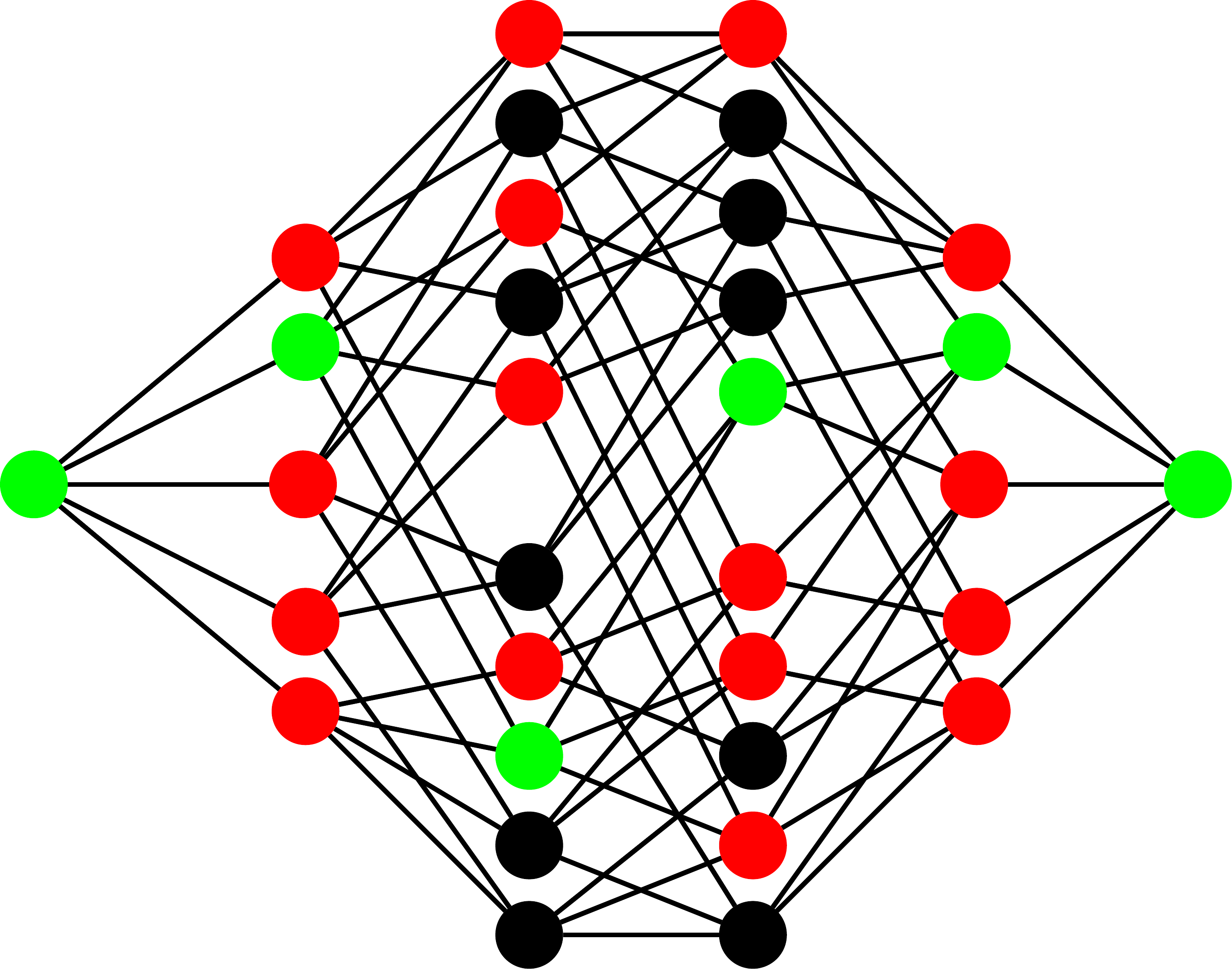}
    \label{fig:scalinga}
    }
    \qquad\subfloat[$n=6$]{
    \includegraphics[width=50mm]{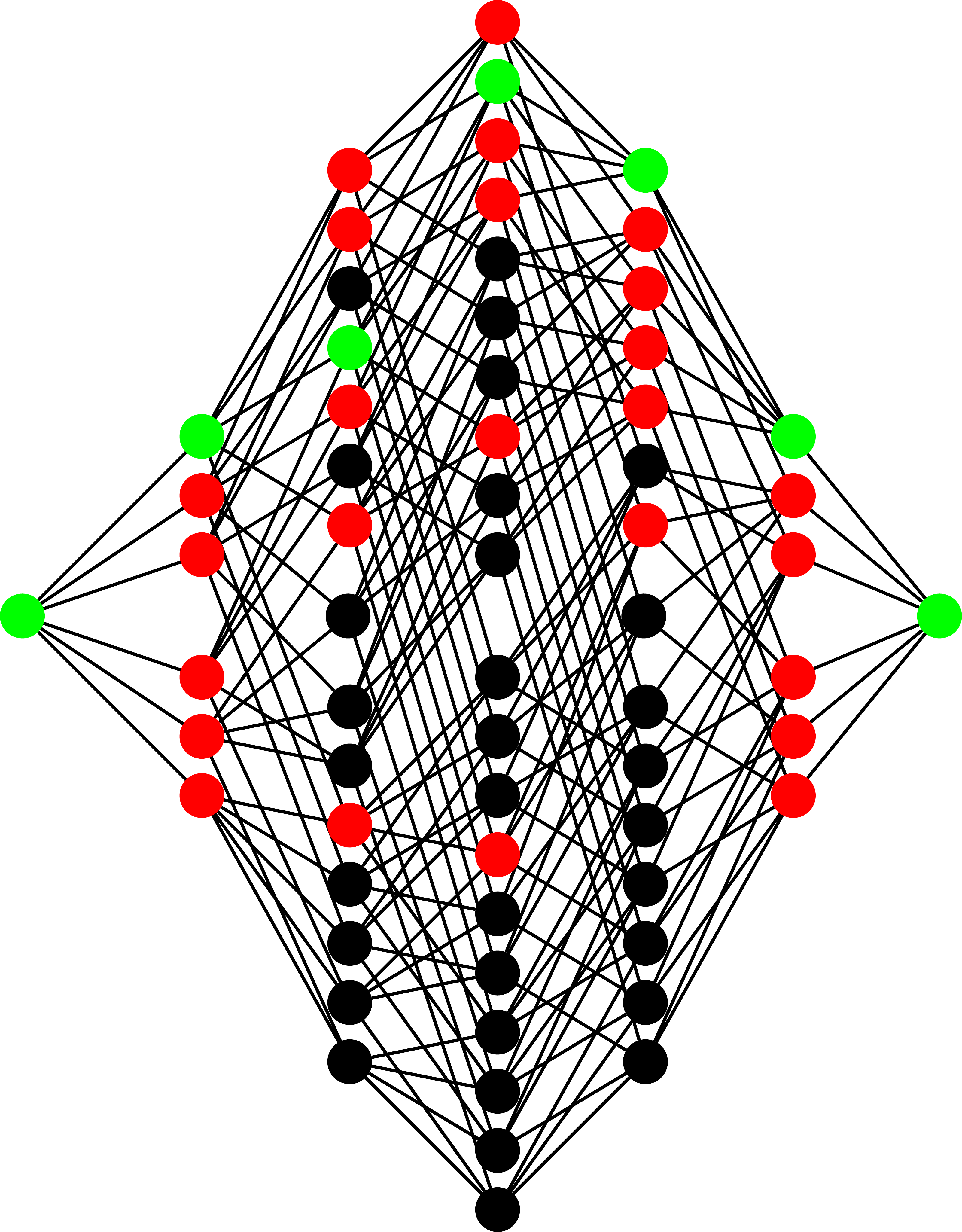}
    \label{fig:scalinga}
    }
    \qquad\subfloat[$n=7$]{
    \includegraphics[width=50mm]{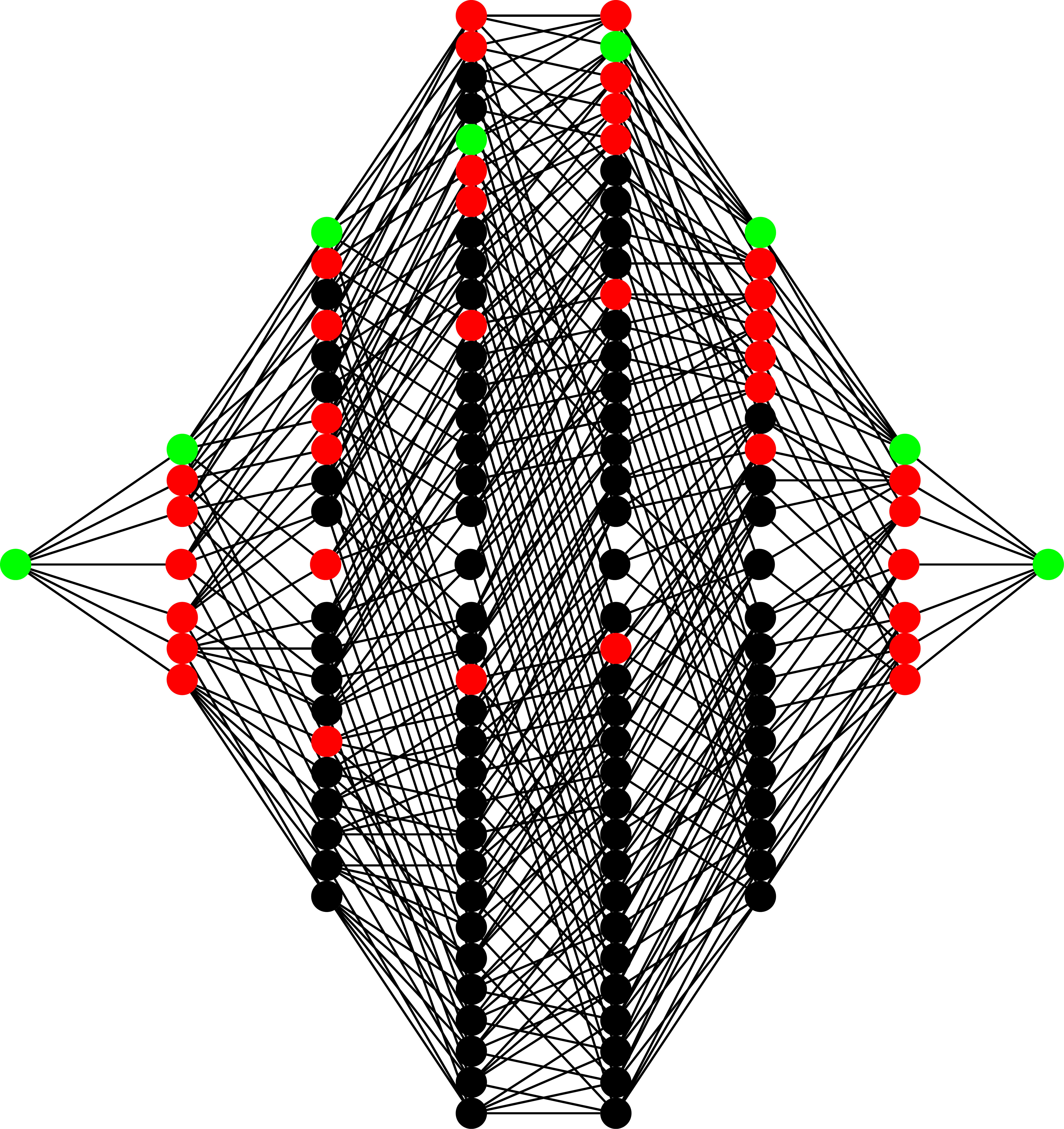}
    \label{fig:scalinga}
    }
    \caption{Visited vertices of $\hg$ in \textproc{FindEffectiveVertices} for Hamming symmetric example on a hypercube for $n=5,6,7$ as in \cref{eq:hamsym}. Green vertices are effective vertices, red vertices are checked, but not added to the set of effective vertices, and black vertices are unchecked. Columns of vertices are in Hamming weight order. Note the number of vertices grows exponentially in $n$, but the number of checked vertices grows only polynomially.}
    \label{fig:scaling}
\end{figure*}

Once we have our effective set of vertices we proceed to calculate the weights $\Omega_{uv}$ in the main loop of Algorithm 2. One can walk through this loop by hand with ease obtaining the matrix

\[
\Omega=\begin{pmatrix}
0 & 2 & 0 & 1\\
2 & 0 & 1 & 0\\
0 & 3 & 0 & 0\\
3 & 0 & 0 & 0
\end{pmatrix}
\]
where the rows and columns are in the same order as listed above in $V'$. Given $\Omega$ and $V'$ we can calculate 

\[
H'=\begin{pmatrix}
-1 & -2 & 0 & -1\\
-2 & 1 & -1 & 0\\
0 & -3 & 3 & 0\\
-3 & 0 & 0 & -3
\end{pmatrix}
\]
from which we can compute the ground state in the symmetric subspace, $\phi'$. We can also compute the size of each equivalence class using \cref{eqn:numclasses} from the main text (repeated here)

\begin{align}
    \frac{\abs*{\eqc{u}}}{\sum_{v \in V'}\abs{\eqc{v}}} =\left(\sum_{v \in V'} \prod_{e\in P(u,v)} \frac{\omega_{e_0\eqc{e_1}}}{\omega_{e_1\eqc{e_0}}}\right)^{-1}
\end{align}
which yields

\begin{align*}
    \abs{\eqc{X_{100}}}&=\abs{\eqc{X_{000}}}=3\\
    \abs{\eqc{X_{010}}}&=\abs{\eqc{X_{101}}}=1.
\end{align*}

Finding the ground state eigenvector of $H'$ gives $\phi'=(3+2 \sqrt{2}, 1+\sqrt{2}, 1, 7 + 5\sqrt{2})$. Now, we compute the probability of sampling each equivalence class using equation \cref{eqn:prob} from the main text (repeated here),

\begin{equation}
    {\Pr{\eqc{u}} = \abs*{\eqc{u}}\phi(u)^2 = \frac{\abs*{\eqc{u}}\phi'(\eqc{u})^2}{\sum_{v \in V'} \abs*{\eqc{v}}\phi'(\eqc{v})^2}.}
\end{equation}

That is, 

\begin{align*}
    \mathrm{Pr}(\abs{\eqc{X_{100}}})&\approx0.32 \quad \mathrm{Pr}(\abs{\eqc{X_{000}}})\approx0.055\\
    \mathrm{Pr}(\abs{\eqc{X_{010}}})&\approx0.003 \quad \mathrm{Pr}(\abs{\eqc{X_{101}}})\approx0.622.
\end{align*}

We then return a member of $V'$ according to the above-stated probability distribution. We can easily verify that this agrees with the probability that, upon computational-basis measurements, we return a member of the corresponding equivalence class of the full Hamiltonian.

In addition to this minimal full example we also provide a depiction of a more complicated clausal theory graph that includes all the different gadget described in the main text (see \cref{fig:fullExample}). Note that this example is not stoquastic so that all the different types of gadgets can be shown clearly on one graph. Thus this example is for illustrative purposes only.

\begin{figure}[h]
    \includegraphics[width=80mm]{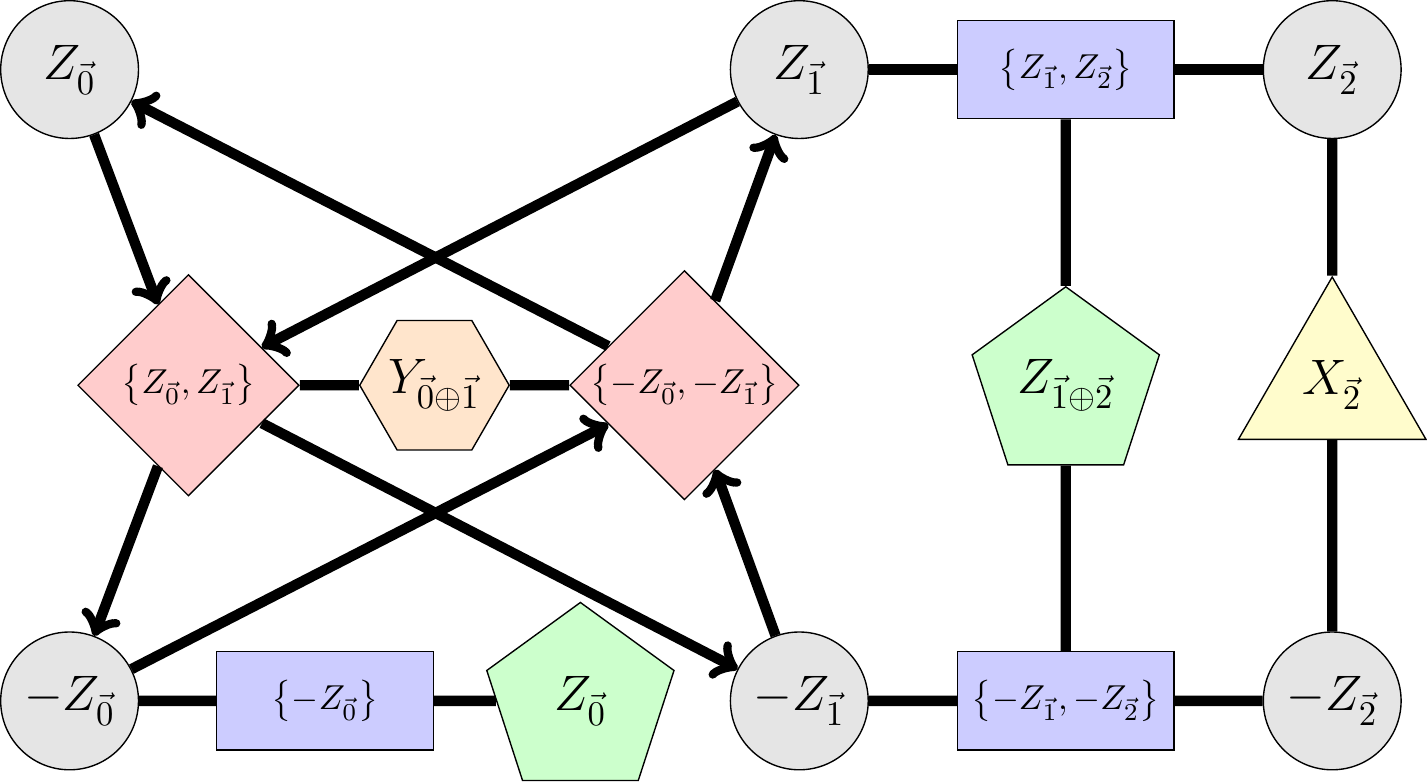}
    \caption{Full example of $M_1(E_\hg)$ for $H=-X_{001}-Y_{110}-Z_{100}+Z_{011}$. Note, $H$ is not stoquastic as to have a legible graph with all relevant types of vertices/edges.}
    \label{fig:fullExample}
\end{figure}

\section{S2: Additional Gadgets}\label{app:summaryofgadgets}
\paragraph{\textbf{Generality:}}In addition to the gadgets corresponding to individual $X_b, Y_b, Z_b$ terms of the Hamiltonian one must also consider additional composite gadgets to achieve full generality for stoquastic Hamiltonians with symmetries between ``like" terms -- i.e. those with interactions between the same number of qubits. Such a gadget is constructed as follows. We note that for each double-typed term of the form $X_b Y_{b'}, X_b Z_{b''}$, or $Y_{b'}Z_{b''}$ appearing in $H$, one would need to introduce re-colored versions of the gadgets $G_1(b), G_2(b')$, and $G_3(b'')$ in order to avoid conflicting with other terms that might appear independently, such as $X_b$. Similarly, for each term of the form $X_b Y_{b'} Z_{b''}$ one would need to introduce yet another color to avoid conflicting with single-typed and double-typed terms. Alternatively, as long as one is consistent for all terms, if $X_b Y_{b'}$ appears and $G_1(b)$ is already included in the clausal theory graph, one might just adjoin an indicator vertex with a unique color to indicate that both single- and double-typed terms exist. See \cref{fig:compgadget}. 

\begin{figure}
    \centering
    $G_1(b)\cup G_2(b')\cup G_3(b'')$$\cup\left(\vcenter{\hbox{\includegraphics[width=25mm]{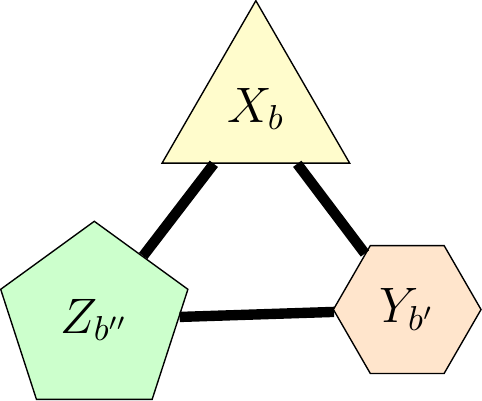}}}\right)$ 
    \caption{Construction of a composite gadget for a Hamiltonian term of the form $X_bY_{b'}Z_{b''}$. $G_1(b), G_2(b')$, and $G_3(b'')$ are as defined in the main text (see Table 1). }
    \label{fig:compgadget}
\end{figure}

Furthermore, as noted in the main text, one could imagine similar constructions that allow for the detection of even more complicated symmetries including those between ``non-like" Hamiltonian terms. While the current work presented here covers all cases previously described in the literature, it is of interest to consider these more general constructions. This combined with a full accounting for algorithms to deal with approximate symmetries would serve to fully generalize the work here. We leave details of such constructions to forthcoming work \cite{bringewattjarretinprep}.

\paragraph{\textbf{Optimality:}} The gadgets in this paper were chosen for simplicity, but are suboptimal from the perspective of GI. For instance, one could combine the gadgets for $XX$ ($G_1(11)$) and $YY$ ($G_2(11)$) into a re-colored version of $G_2(11)$. One can similarly include $ZZ$ terms in a gadget of this form. (See the $Y_{110}$ and $Z_{011}$ gadgets of \cref{fig:fullExample} and note that the structure of a potential $Z_{110}$ gadget would also be captured by $Y_{110}$ by deleting edges pointing towards literals $\pm Z_{\vec{0}}$ and $\pm Z_{\vec{1}}$.)

If one wishes to make the algorithm as efficient as possible, one would ideally choose such a minimalist construction.

\section{S3: Proofs}\label{app:proofs}
\begin{proof}[Proof of \cref{thm:M_inverse}]
    By construction $M:\hg \mapsto \ctg$ is unique, so we only need to show that the $M^{-1}$ exists. First, note that $M_0:V_\hg \longrightarrow M_0[V_\hg] \subset V_\ctg$ is bijective and provides a unique mapping from vertices to assignments of literals. Hence, we only need to show how to derive $w(u,v)$ from $\ctg$ for all $(u,v) \in V_\hg \times V_\hg$.
    
    For any $u \in V_\hg$, we calculate $w(u,\infty)$ by considering all shortest paths $P_b$ connecting $A(u)^{(\mathpzc{b})}$ to each clause cluster vertex $Z_b^{(\mathpzc{z}_b)}$ through $\ctg$. In particular we obtain  
    \begin{equation}
        w(u,\infty) = 2\sum_{b:\abs{P_b} \neq \hw*{b}}\abs{\kappa_b} - \sum_b \abs{\kappa_b}
    \end{equation}
    where we recall that we can extract $\abs{\kappa_b}$ from the color $\mathpzc{z}_b$ of each clause cluster vertex. 
    
    Similarly, we can consider any pair $(u,X_bu) \in V_\hg \times V_\hg$. We consider $\mathcal{M}(u,X_b u) = \left(L \cup -L\right)\setminus \left(M_0(u) \cap M_0(v)\right)$ and construct $S = \mathcal{M}(u,X_bu)\cup M_1(\{u,X_bu\})$. Let $P_b$ be the number of shortest paths connecting $A(u)^{(\mathpzc{b})}$ to $A(X_b u)^{(\mathpzc{b})}$ through $S$. Then, 
    
    \begin{equation}
    w(u,X_b u) = \begin{cases}
    \alpha_b + (-1)^{\frac{\abs*{P_b}}{\hw{b}^2}}\beta_b & \mathrm{if} G_1(b)\cup G_2(b) \subset \ctg \\
    0 & \mathrm{otherwise}
    \end{cases}
    \end{equation}
    
    Hence, $\exists\, M^{-1}:\ctg \mapsto \hg$ and $M$ is bijective.
\end{proof}

Recall $\ctg(u) = M_0(u) \cup M_1(E_\hg)$. We now provide a straightforward lemma that we use to prove \cref{thm:automorph}. In the below proofs we drop the ${(\mathpzc{b})}$ superscript on assignment vertices to avoid notational clutter.

\begin{lemma}\label{thm:assign_change}
There exists a color-preserving automorphism $g$ of $\ctg$ with $g \circ A(u) =  A(v)$ iff $\ctg(u)\simeq\ctg(v)$.
\end{lemma}

\begin{proof}
    When $u = v$, this is trivial. We now assume $u \neq v$.

    Suppose $g$ is an automorphism of $\ctg$ with $g \circ A(u) = A(v)$. Then, clearly, $\restr{g}{V_{\ctg(u)}}:V_{\ctg(u)}\longrightarrow V_{\ctg(v)}$ is a color-preserving isomorphism between $\ctg(u)$ and $\ctg(v)$.
    
    Conversely, suppose there exists a color-preserving isomorphism $\restr{g}{V_{\ctg(u)}}:V_{\ctg(u)}\longrightarrow V_{\ctg(v)}$ such that $\ctg(u)\simeq\ctg(v)$. We note that by construction $\ctg(u) \cap \ctg(v) = M_1(E_\hg)$ and color-consistency requires $\restr{g}{V_{\ctg(u)}}\circ A(u)=A(v)$. Therefore, $\restr{g}{V_{\ctg(u)}}$ is an automorphism of $M_1(E_\hg)$. Extend $\restr{g}{V_{\ctg(u)}}$ to $g:V_\ctg \longrightarrow V_\ctg$ such that $g\circ A(X_b) = \set*{g\left((-1)^{b_i}Z_{\vec{i}}^{(\mathpzc{a})}\right)}_{i=0}^{n-1} = A(X_{b'})$. Then, $g$ is an isomorphism between $M_0(X_b)\cup M_1(E_\hg)$ and $M_0(X_{b'})\cup M_1(E_\hg)$ for any choice of $b$. Hence, $g$ is an automorphism of $\ctg$.
   
\end{proof}

Now, we can prove \cref{thm:automorph}.
\begin{proof}[Proof of \cref{thm:automorph}]
    Suppose $u \equiv f(u) \in V_\hg$. By definition, $f\in \aut(\hg)$. By construction, we have that $g$ is a color-preserving isomorphism $\ctg(u) \simeq \ctg(f(u))$.
    
    Now, suppose that $\restr{g}{V_{\ctg(u)}}$ is a color-preserving isomorphism $\ctg(u) \simeq \ctg(v)$.  Then, by \cref{thm:assign_change}, there exists some color-preserving automorphism $g:V_{\ctg} \longrightarrow V_{\ctg}$ of $\ctg$. Define $f=A^{-1}\circ g \circ A$. Then, by \cref{thm:M_inverse}, $f$ is an automorphism of $\hg$ and, thus, $u \equiv f(u)$. 
\end{proof}

\section{S4: Smooth Transitions}\label{app:transitions}
We apply the result of \cite{jarret2018hamiltonian} under very weak constraints to show that families of Hamiltonians almost invariably encounter an exponentially small gap, unless they undergo very smooth phase transitions. These results are similar to but, in terms of gap-analysis, stronger than those of \cite{farhi2008make}. Here, we reference only the behavior of the ground state and show that most phase transitions, like those we expect out of adiabatic optimization, produce exponentially small gaps. The following simple theorem is sufficiently illustrative, although its statement could be improved asymptotically and easily generalized to include more than $k$-local Hamiltonians.

\renewcommand{\thethm}{S.\arabic{thm}}
\begin{thm}\label{thm:smooth}
 If $H$ is a $k$-local Hamiltonian with ground state $\phi$, $\norm{H}\leq 1$, and there exists a set $S_0$ such that
 \begin{enumerate}
     \item $S_0 = \set*{u \;\vert\; \abs{\phi(u)} < 2^{-n^c}}$ with absolute constant $c>0$,
     \item $\sum_{u \in S_0} \abs*{\phi(u)}^2,\sum_{u \notin S_0 } \abs*{\phi(u)}^2 = \Omega\left(\frac{1}{\poly(n)}\right)$,
     \item and $\abs*{\overline{S_0}} = \bigO{\poly(n)}$;
 \end{enumerate} then $\gamma(H) = 2^{-\Omega(n^c)}$.
\end{thm}
\begin{proof}
    Note that by \cite{jarret2018hamiltonian}, the weighted Cheeger constant $h$ bounds $\gamma(H)$, as $2h \geq \gamma(H)$. In particular, $h = \min_{S} h_S$ where $h_{S}$ is the weighted Cheeger ratio
    \begin{align*}
        h_S &= \frac{\sum_{u \in S, v \notin S}\Re\left(-H_{uv}\phi(u)\phi(v)\right)}{\min\left\{\sum_{u \in S}\abs*{\phi(u)}^2,\sum_{u \notin S}\abs*{\phi(u)}^2 \right\}}. 
    \end{align*}
    Now, consider $S_0$,
    \begin{align*}
        h_{S_0} &\leq \bigO{\frac{\max_{v_0 \in S_0}\abs*{\phi(v_0)} \sum_{u \in S_0, v \notin S_0}\abs*{-H_{uv}}\abs*{\phi(u)}}{\min\left\{\sum_{u \in S_0}\abs*{\phi(u)}^2,\sum_{u \notin S_0}\abs*{\phi(u)}^2 \right\}}} \\
        &= \bigO{\frac{\poly(n)\max_{v \in S_0}\abs*{\phi(v)}}{\min\left\{\sum_{u \in S_0}\abs*{\phi(u)}^2,\sum_{u \notin S_0}\abs*{\phi(u)}^2 \right\}}} \\
        &= \bigO{\frac{\poly(n)2^{-n^c}}{\min\left\{\sum_{u \in S_0}\abs*{\phi(u)}^2,\sum_{u \notin S_0}\abs*{\phi(u)}^2 \right\}}} \\
        &= 2^{-\Omega(n^c)}.
    \end{align*}
    Since the weighted Cheeger constant $h = \min_{S} h_{S} < h_{S_0} = 2^{-\Omega(n^c)}$, $\gamma(H) = 2^{-\Omega(n^c)}$.
\end{proof}
Thus, if we are interpolating over a family of Hamiltonians $H(s)$ and we ever encounter a ground state $\phi$ such that (1) there exists a set $S$ where we have substantial probability of returning a sample from either $S$ or $\overline{S}$, (2) for any $u \in S$ it is unlikely that we will sample $u$ in time $\bigO{\poly(n)}$, and (3) $\overline{S}$ is small, we encounter an exponentially small gap. Since we typically interpolate over a family of Hamiltonians $H(s)$ such that the ground state $\phi_0$ of $H(0)$ has $\abs*{\phi_0(u)}^2 = \bigO{2^{-n/2}}$ for any $u$ and end in a Hamiltonian $H(1)$ such that the ground state $\phi_1$ satisfies $\sum_{u \in S}\abs*{\phi_1(u)}^2 = \Omega\left(1/\poly(n)\right)$ for some small set of computational basis states $\abs*{S} = \bigO{\poly(n)}$, avoiding the constraints of \cref{thm:smooth} with naive families $H(s)$ is unlikely.

\end{document}